\documentclass[runningheads]{llncs}
\usepackage{macros-data-mining}
\usepackage[normalem]{ulem}
\begin{document}
\title{From Community Detection \\ to  Community Deception}
\author{Valeria Fionda\inst{1} \and Giuseppe Pirr\`o\inst{2}}
\authorrunning{V. Fionda \and G. Pirr\`o}   
\institute{DeMaCS, University of Calabria, Italy\\
	\email{fionda@mat.unical.it}\\
	\and
	Institute for High Performance Computing and Networking, ICAR-CNR, 
	Italy\\
	\email{pirro@icar.cnr.it}
}
\tocauthor{
	Valeria Fionda (University of Calabria)
	Giuseppe Pirr\`o (ICAR-CNR)
}
\maketitle
\begin{abstract}
The community deception problem is about how to hide 
a target community $\comH$ from community detection algorithms. The 
need for deception emerges whenever a group of entities (e.g., activists, 
police enforcements) want to cooperate while concealing their existence as a 
community. In this paper we introduce and formalize the community 
deception problem.  To solve this problem, we describe algorithms that 
carefully rewire the connections of $\comH's$ members. 
We experimentally show how several existing
community detection algorithms can be deceived, and quantify the level of 
deception by introducing a deception score. 
We believe that our study is intriguing since, while showing how deception 
can be realized it raises awareness for the design of novel detection 
algorithms robust to deception techniques.
\end{abstract}
\section{Introduction}
\label{sec:introduction}
Many aspects of everyday life involve networks; social networks, biological 
networks, and the World Wide Web are just a few examples. The study of 
networks touches many disciplines ranging from physics to 
computer and social science. 
One important task in network analysis is the identification of 
communities, that is, regions (subsets of vertices) of a network that help to 
gain insights about its structure~\cite{fortunato2010community}.
Detecting communities is useful for several purposes such as identifying 
topics in information networks~\cite{revelle2015finding}, 
criminal organization from mobile networks~\cite{ferrara2014detecting}, 
friendship in social networks~\cite{yang2013community} or motifs in 
biological networks~\cite{fionda2011biological}.

While community detection 
is a well-understood and studied problem, little has been done in terms of 
\textit{community deception}. 
Solving the community deception problem amounts at devising 
techniques to conceal the existence of a target community from community 
detection algorithms. Studying community deception is intriguing from two 
different perspectives. On one hand, deception techniques can be useful for 
activists in despotic regimes to hide themselves (as a group) from software 
like the Laplace's 
Demon, a protest monitoring system developed by a pro-Kremlin 
group~\cite{laplace-daemon};
or, police enforcements to avoid to be tracked as done by Ukrainian bloggers 
that tracked Russian soldiers on social media~\cite{ukranian-bloggers}. 
On the other hand, the study of deception techniques raises awareness 
for design of novel community detection algorithms robust to 
community deception techniques as deception could also be used for 
malicious purposes. 
When embarking on the study of community deception we 
identified some research challenges, among which: (i) how to practically 
realize community deception? (ii) how to devise computationally 
feasible algorithms? (iii) how to assess the degree of deception of  
a target community $\comH$? We show how to tackle challenge (i) by 
rewiring in a principled way the connections of $\comH'$s members. As for 
challenge (ii) we present two greedy algorithms; the first based on 
modularity and the second one based on a novel measure of community 
safeness. To 
tackle challenge (iii) we introduce a deception measure that computed 
before and after applying community deception algorithms allows to 
measure their success.

\smallskip
\noindent
\textit{Related Work.} Community detection algorithms strive 
to maximize cluster quality measures such as 
modularity~\cite{blondel2008fast,newman2006finding}, adopt probabilistic 
approaches based, for instance, on random 
walks~\cite{yang2013community,rosvall2008maps} or use network 
attributes~\cite{yang2013community}. Yet other approaches study 
the problem of finding a community given a set of 
vertices~\cite{barbieri2015efficient}. 
Fortunato~\cite{fortunato2010community} provides a comprehensive study 
on this topic while other studies focus on the evaluation of community 
detection algorithms (e.g.,~\cite{leskovec2010empirical,yang2015defining}).

In this paper we take a different direction and tackle the 
problem of designing algorithms to \textit{deceive} community detection 
algorithms. Our goal is to to hide a community from being discovered 
by community detection algorithms. We are not aware of any previous work 
on this topic.
Note that community deception differs from community preservation.
This latter problem is usually tackled via techniques such as 
k-anonymity, k-degree 
anonymity~\cite{campan2015preserving} or 
k-isomorphism~\cite{cheng2010k} and is focused on the 
assessment of how well the anonymization preserves  communities from the 
original network~\cite{campan2015preserving}. In contrast, 
 tackling community deception does not require anonymization as the 
 goal is to hide a community while keeping its identity (i.e., the identity of its 
 members) untouched.

\smallskip
\noindent
\textit{Contributions and Outline.} We make the following main 
contributions: 
\vspace{-.2cm}
\begin{itemize}
	\item[$\bullet$] introducing and formalizing the community deception 
	problem, which, to the best of our knowledge, has no been studied 
	before; 
	\item[$\bullet$] presenting two algorithms for community deception, one 
	based on modularity and the other based on a novel measure of 
	community safeness;
	\item[$\bullet$] showing how our algorithms are able to deceive several 
	existing community detection algorithms on real and synthetic networks.
\end{itemize}

\noindent
The remainder of the paper is organized as follows. 
Section~\ref{sec:preliminaries} introduces the problem and provides an  
example. Section~\ref{sec:min-modularity} 
presents our first algorithm for community deception based on 
modularity. In Section~\ref{sec:safeness} we introduce our second algorithm 
based on community safeness. The experimental 
evaluation is discussed in Section~\ref{sec:evaluation}.
We conclude and sketch future work in Section~\ref{sec:conclusions}.
%
\vspace{-0.3cm}
\section{Problem Statement and Running Example}
\label{sec:preliminaries}
%

\noindent
A network $\netD$ is an undirected graph that includes a set of 
\numnodes:=$|\nodes|$ vertices and \numedges:=$|\edges|$ edges.
 We denote by 
 $\degV{v}$=$|\neigh{v}|$ the degree of $v$, where $\neigh{v}$ is the set of neighbors of $v$. The set of communities is 
 denoted by
 $\comS$=\{$\com_1,\com_2,...\com_k$\}
 and $\com_i\in \comS$ 
		denotes the \textit{i-th} community.	
	$E(\com_i)$ denotes the set of edges that are incident to some nodes in 
		$\com_i$. We distinguish between intra-community edges of the 
		form $\eC{u}{v}: 
		u,v\in \com_i$ and inter-community edges of the form 
		$\edgesCs{u}{v}$:$u\in \com_i, v\in \com_j$.
		The 
		degree of a community is denoted by: $\degC{\com_i}$=$\sum_{v\in 
			\com_i} \ \degV{v}$. $\netPE$ 
		(resp., $\netME$) denotes a set of edge additions (resp., deletions) 
		on 
		$\net$. 
		We denote by  $\comH\subseteq \nodes$ the community, not 
		necessarily part of $\comS$, that we want to hide from community 
		detection algorithms.
\begin{probl}[\textbf{Community Deception}]	Let $\netD$ be a 
	network 
	and $\comAlgo$ a community 
	detection algorithm. Given a community $\comH \subseteq 
	\nodes$, and a deception function $\phi_{\comAlgo} (\net, \netN)$, find 
	a 
	network 
	$\netND$ with $E'=(E\cup \netPE) \setminus \netME$ such that:
	\begin{center}
		\vspace{-.9cm}
		\begin{equation}
		\label{eq:comProtProblem} \argmax_{\netN}\{ \phi_{\comAlgo} (\net, 
		\netN)\}\end{equation}

		\vspace{-.34cm}
		{where \small $\netPE$$\subseteq$$\{(u,v): u$$\in$$\comH \vee v$$\in$$\comH, 
			(u,v)$$\notin$$E\}$ and $\netME$$\subseteq$$\{(u,v): 
			u$$\in$$\comH \vee v$$\in$$\comH, (u,v)$$\in$$E\}$\Big\}}
	\end{center}
	\label{def:probl-com-prot}
\end{probl}

%
\vspace{-.1cm}

The function $\phi_{\comAlgo}$ mimics the process of \textit{deceiving} a 
community 
detection algorithm $\comAlgo$ so that $\comH\notin \comS$. 
Solving the community deception problem amounts at designing algorithms 
capable to find an updated network $\netND$ so that $\phi(\net,\netN)$ is 
maximized. An optimal algorithm for this problem is computationally 
hard as it requires an exhaustive exploration of all possible 
combinations of (subsets of) edge updates (i.e., $\netPE$ and 
$\netME$). As we will describe shortly, we resort to greedy 
algorithms that find the local optimum at each evaluation step.

Moreover, to measure the level of deception of $\comH$ we introduce the 
$\hScore{}$ 
score, which encompasses different kinds of information such as reachability 
between $\comH's$ members and their spreading in 
$\comSD$. This score (see 
Definition~\ref{def:deception-score}),
computed before ($\hScore{\net}(\comH,\comAlgo)$) and after 
($\hScore{\netN}(\comH,\comAlgo)$) the usage of a community deception algorithm, allows to 
quantify its performance -- we will write $\hScore{\netN}$ when $\comH$ and $\comAlgo$ are clear from the context. The worst case (i.e., $\hScore{}$=0) occurs 
when $\comH$ belongs to the output of $\comAlgo$ ($\comH\in 
\comS$). If $\comH\notin \comS$ then $\comH's$ members can 
be spread inside $\comS$ in many ways, thus leading to many $\hScore{}$ 
values. The best $\hScore{}$ (i.e., $\hScore{}$$\sim$1) is obtained when 
$\comH's$ members 
are reachable from one another and spread in different (large) communities.

\smallskip
\noindent
\underline{Running Example.} Consider the Zachary's Karate Club 
network~\cite{zachary1977information} and the partition in communities 
$\comS=\{\com_1,\com_2,\com_3,\comH\}$  shown in 
Fig.~\ref{fig:running-example} (a) obtained by the Louvain community 
detection algorithm ($\comAlgo$=\texttt{louv})~\cite{blondel2008fast}. To model 
the 
worst-case scenario from the deception point of view, we assume 
$\comH$=$\{24,25,26,28,29,32\}$, that is, 
$\comH\in\comS$ and then $\hScore{\net}(\comH,\comAlgo)$=0. We now 
outline our  
deception algorithms.

In the first algorithm ($\algoMod$) the function $\phi_{\comAlgo}$ (in the 
general 
statement of the deception problem) is the
\textit{modularity loss} $\modLoss$=$\modCH$-$\modCHS$. $\algoMod$'s 
greedy strategy at each step picks the edge change with the 
highest $\modLoss$. Our choice to use modularity for 
community deception stems from the observation that several community 
detection algorithms (e.g.,~\cite{blondel2008fast,newman2006finding}) are 
based on modularity maximization. Therefore, if the edge update found 
by $\algoMod$ introduces a modularity loss, then a 
community detection algorithm applied to $\netN$ (i.e., $\net$ after 
applying the update) will possibly give a partitioning in communities more 
favorable to $\comH$ than that in $\net$ (i.e., 
$\hScore{\netN}>\hScore{\net}$). 
\begin{figure}[!t]
	\centering
	\includegraphics[width=\textwidth]{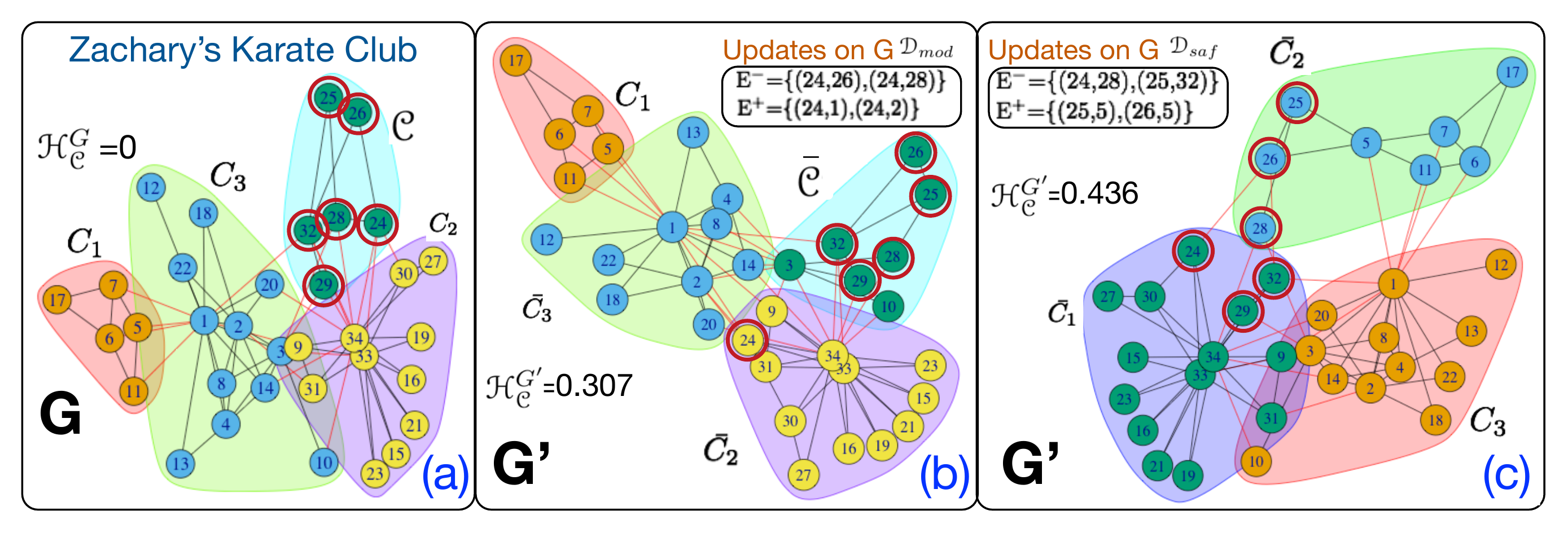}
	\vspace{-.7cm}
	\caption{Communities found by Louvain~\cite{blondel2008fast}(a); 
		output of Louvain after modularity-based deception (b); 
		output of Louvain after safeness-based deception (c).}
	\label{fig:running-example}
	\vspace{-.56cm}
\end{figure}
%
One key feature of $\algoMod$ is that to determine the best edge updates it 
does not require to recompute modularity from scratch for each candidate 
edge. $\algoMod$ leverages updating rules (see 
Section~\ref{sec:impact-edges-modularity}) able to measure the impact of 
an update on modularity before applying it. Fig.~\ref{fig:running-example} 
(b) reports the output of the community detection algorithm $\comAlgo$ on 
$\netN$, the network obtained after applying the updates found by 
$\algoMod$ (reported in the inner-box) on $\net$.
In terms of deception, the situation for $\comH$ has improved. This is 
because: (i) 
its members are now spread in two communities (e.g., 24, 
is now part of $\bar{\com_2}$); (ii) other nodes of $\comH$ are now 
grouped with node 3 and 10 in $\bar{\comH}$. Indeed, the deception score 
goes from $\hScore{\net}$=0 to $\hScore{\netN}$=0.307. Nevertheless, 
node 24 is now disconnected from the other members of 
$\comH$.
It is worth to mention that the choice of the type of update is 
subtle. In 
fact, if one were to add the edge (24,25) in 
Fig.~\ref{fig:running-example} (a) instead of (24,1) and running again 
$\comAlgo$ on the updated network, modularity would have 
increased and all $\comH's$ members would have remained 
in the same community;
same reasoning for the deletion of (24,33) 
instead of 
(25,26). 
We will formally study these aspects in Section~\ref{sec:min-modularity}.

Since not all detection algorithms are based on modularity 
(e.g.,~\cite{rosvall2008maps}) we have devised another deception algorithm 
$\algoSaf$  where $\phi_{\comAlgo}$ is the \textit{safeness gain} 
$\sIncrease$=$\safScoreN{\netN}-\safScore$. Safeness $\safScoreN{}$ 
looks 
at reachability 
between $\comH's$ members and their connection with nodes not in $\comH$ (see Section~\ref{sec:safeness}). Also in this case it is possible to  
determine the 
impact of updates on safeness before applying them to $\net$ 
(Section~\ref{sec:impact-edges-safeness}).
The output of $\comAlgo$ after applying the changes (see inner-box) 
detected by $\algoSaf$ is reported in Fig.~\ref{fig:running-example} (c). 
Intuitively,  $\algoSaf$ gives a better set of changes than $\algoMod$ since: 
(i) 
$\comH's$ members are now equally spread in two communities while in 
Fig.~\ref{fig:running-example} (b) this is not the case; (ii) $\comH's$ 
members are now better ``hidden'' with nodes in $\bar{\com_1}$ and 
$\bar{\com_2}$; (iii) all nodes of $\comH$ are reachable from one another. 
$\algoSaf$ gives a higher deception score, that is, $\hScore{\netN}$=0.436.
%

The worst-case scenario discussed in this example underlines how 
our community deception algorithms were able to detect a few 
updates that significantly increased the deception of $\comH$ in a real 
network. Even if in this example the Louvain algorithm was used, 
our algorithms can deceive any community detection algorithm as we will 
discuss in the experimental evaluation section (Section~\ref{sec:evaluation}). 
%
\section{Community Deception via Modularity}
\label{sec:min-modularity}
We now introduce the first community deception algorithm 
($\algoMod$) based on modularity~\cite{newman2006modularity},  a 
well-studied measure\footnote{Other types of 
modularity (e.g., generalized~\cite{ganji2015generalized}) are 
orthogonal to our study.} in the community detection literature.
\vspace{-.4cm}
\begin{defn} \normalfont \textbf{(Modularity).} Given a network $\net$, the 
	modularity of the partition of this network into communities 
	$\comS$=\{$\com_1,\com_2,...\com_k$\} is given by:
	\vspace{-.2cm}
	\begin{equation}
	\modularity{\net}{\comS}=
	\frac{\eta}{\numedges} - 
	\frac{\delta}{4\numedges^2}
	\label{eq:modularity}
	\vspace{-.1cm}
	\end{equation}
	
where $\eta$=$\sum_{\com_i \in 
	\comS}{|\edgesC{\com_i}|}$ and $\delta$=$\sum_{\com_i \in 
	\comS}{\degC{\com_i}}^2$. 
\end{defn}


Modularity measures the number of edges falling within groups minus their 
expected number in an equivalent network with edges placed at random
The objective of many community detection algorithms is to 
maximize modularity~\cite{brandes2008modularity}. 
Our first deception algorithm $\algoMod$ 
considers the function $\phi_{\comAlgo}$ (see 
Problem~\ref{def:probl-com-prot}) to be the 
modularity loss 
$\modLoss$=$\modCH-\modCHS$ and thus the goal is to find the set of 
edge updates where $\modLoss$ is maximized.
Newman~\cite{newman2006finding} touched the somehow related problem 
of modularity minimization for the discovery of 
anti-communities~\cite{newman2006finding}.
Our approach differs in two main respects. First, community deception strives 
to maximize $\modLoss$ w.r.t. $\comH$, that is, via edge updates 
performed by $\comH's$ members only. Second, Newman's study did not 
report on the impact of the different types of edge updates 
on modularity while we formally tackle this problem in 
Section~\ref{sec:impact-edges-modularity}.

As anticipated in the running example, $\algoMod$ adopts a {greedy 
strategy} that at 
each step 
identifies the edge update that brings the highest modularity loss. 
In what follows, we first study the modularity loss for the different types of 
edge updates and then outline the $\algoMod$ 
algorithm. 
\subsection{Impact of Edge Updates on Modularity}
\label{sec:impact-edges-modularity}
Let \netD\ be a network and $\comSD$ a partitioning 
having modularity $\modularity{\net}{\comS}$. Let 
$\modLoss$=$\modularityNet$-$\modularityNetN$ be the modularity loss
and $\comH\subseteq\nodes$ a community. 

\smallskip
\noindent
{\underline{Edge Addition}.}  We first consider the
addition of an inter-community edge. 
\begin{theorem} For any 
	inter-community edge addition $(u,w)$: $u \in C_i\cap 
	\comH, w \in \com_j$, with $i\neq j$ giving {\small 
	$\netN=(\nodes,\edges\cup\{\e{u}{w}\})$} we have that:
	\begin{center}$\modLoss>0$ if, and only if, $\frac{\eta}{m(m+1)}+\frac{2m^2(\degC{C_i}+\degC{C_j}+1)-\delta(2m+1)}{4m^2(m+1)^2}>0$.
	\end{center}
	\label{th:inter-edge-add}
\end{theorem}
\begin{proof}
By manipulating eq.~(\ref{eq:modularity}) we 
have that $\eta$ (sum of edges within communities) remains 
unchanged while $\delta$ (sum of the degrees in all communities) becomes 
$\delta=\delta+2+2\degC{\com_i}+2\degC{\com_j}$. This gives the new 
value of modularity:
\begin{equation}
\scriptsize
\quad \modularity{\netN}{\comS}= 
\left(\frac{\eta}{\numedges+1}\right) - \left(\frac{\delta+2+2\degC{\com_i}
	+2\degC{\com_j}}{4(\numedges+1)^2}\right)
\label{eq:add-inter-modularity}
\end{equation}
The possible modularity loss is:
{
\scriptsize
$$\modLoss=\modularity{\net}{\comS}-\modularity{\net'}{\comS}=\frac{\eta}{\numedges} - 
	\frac{\delta}{4\numedges^2}-\frac{\eta}{\numedges+1} + \frac{\delta+2+2\degC{\com_i}
		+2\degC{\com_j}}{4(\numedges+1)^2}=$$
		$$=\frac{\eta}{m(m+1)}+\frac{2m^2(\degC{C_i}+\degC{C_j}+1)-\delta(2m+1)}{4m^2(m+1)^2}$$

}

The modularity in $\netN$ that derives from the addition of an 
inter-community edge is independent from $u$ and $w$ as it only depends 
on the degrees of $\degC{\com_i}$ and $\degC{\com_j}$; the higher 
$\degC{\com_i}$ and $\degC{\com_j}$ the higher the modularity loss. The 
maximum loss can be obtained by picking as source and target communities 
for the edge addition the communities having the highest degrees.
\qed
\end{proof}

If $\comH \in \comS$, the possible modularity loss depends on 
the rank (in terms of degree) of $\comH$ in $\comS$. 
To give a hint about the result in Theorem~\ref{th:inter-edge-add}, 
consider the network in Fig.~\ref{fig:running-example} (a) where 
$\comH\in\comS$ and $\degC{\comH}$=24. Note that the 
edge (26,18) identified by $\algoMod$ brings the highest modularity loss 
since 18 is in the community with the highest degree (i.e., 
$\degC{\com_3}$=62).  
If $\comH \notin \comS$, the result of Theorem~\ref{th:inter-edge-add} 
still holds; the edge insertion with 
the highest possible loss is $(u,w)$: $u \in \com_i\cap\comH, w \in \com_j$ and 
$\degC{\com_i}+\degC{\com_j}$ is maximal. We now consider the addition 
of an intra-community edge.
\begin{theorem} For any 
	intra-community edge addition $(u,w)$: $u\in \com_i\cap\comH, w \in \com_i$ giving {\small 
		$\netN=(\nodes,\edges\cup\{\e{u}{w}\})$} we have that:
	\begin{center}
$\modLoss>0$ if, and only if, $\frac{\eta-m}{m(m+1)}+\frac{4m^2(\degC{C_i}+1)-\delta(2m+1)}{4m^2(m+1)^2}>0$. \\
	\end{center}
	\label{th:intra-edge-add}
\end{theorem}
\begin{proof}
	By manipulating eq.~(\ref{eq:modularity}) we 
		have that $\eta'$=$\eta$$+$$1$ and 
		$\delta'$$=$$\delta$$+$$(\degC{\com_i}$$+$$2)^2-\degC{\com_i}^2$ giving the new 
		value of modularity:
		\begin{equation}
		\scriptsize
		{\modularity{\net'}{\comS}}=
		\left( \frac{\eta+1}{\numedges+1}\right) - 
		\left(\frac{\delta+4+4\degC{\com_i}}{4(\numedges+1)^2} \right)
		\label{eq:add-intra-modularity}
		\end{equation}

The possible loss is independent from $u$ and $w$; it 
	only depends on the degree of the community $\com_i$. In this case:
	
{	\scriptsize
	$$\modLoss=\modularity{\net}{\comS}-\modularity{\net'}{\comS}=\frac{\eta}{\numedges} - 
		\frac{\delta}{4\numedges^2}-\frac{\eta+1}{\numedges+1} + \frac{\delta+4+4\degC{\com_i}}{4(\numedges+1)^2}=$$
			$$=\frac{\eta-m}{m(m+1)}+\frac{4m^2(\degC{C_i}+1)-\delta(2m+1)}{4m^2(m+1)^2}$$
}
	\qed
\end{proof}
 
If $\comH=\com_i \in \comS$, then the possible modularity loss deriving 
from an intra-community  edge addition depends on the degree of 
$\com_i$. 
 If $\comH \notin \comS$, Theorem~\ref{th:intra-edge-add} 
 still holds; the edge insertion with the possible highest loss is $(u,w)$: $u \in 
 \com_i\cap\comH, w \in \com_i$ and $\degC{\com_i}$ is maximal.
	By considering an inter-community edge addition between 
	communities $\com_i$ and $\com_j$ (giving the network $\net'$) 
	and an intra-edge addition in the community $\com_i$ (giving the 
	network $\net''$) we have that:
	$\modularity{\net'}{\comS}-	
	{\modularity{\net''}{\comS}}$=$\frac{\degC{\com_i}-\degC{\com_j}
		-2\numedges-1}{2(\numedges+1)^2}$. Since $\degC{\com_i}\leq2m$, we have that $\modularity{\net'}{\comS}-	
			{\modularity{\net''}{\comS}}< 0$ and, thus,
		${\modularity{\net''}{\comS}}>\modularity{\net'}{\comS}$.
\begin{corollary}
The best edge addition, in terms of possible modularity loss, is an inter-community edge between the 
communities $\com_i$ and $\com_j$ having the highest cumulative degree and 
 such that $\com_i\cap\comH\neq \emptyset$. The modularity loss is the 
 same for each edge addition no matter the pair of nodes $u\in\com_i$ 
and $w\in\com_j$.
\label{cor:edge-addition}
	\end{corollary}
%
\noindent
{\underline{Edge Deletion}.} The proofs of the following theorems, similar in 
spirit to those of edge additions, are 
available in the Appendix. We start with the deletion of an 
inter-community edge.
\begin{theorem} For any 
	inter-community edge deletion $(u,w)$: $u \in \com_i\cap\comH, w \in 	
	\com_j$, with $i\neq j$ giving {\small 
	$\netN=(\nodes,\edges\setminus\{\e{u}{w}\})$} we have that:
	\begin{center}
		\small
		$\modLoss>0$ if, and only if,  
		$\frac{\delta(2m-1)-2m^2(\degC{C_i}+\degC{C_j}+1)}{4m^2(m-1)^2}-\frac{\eta}{m(m-1)}>0$ \\
	\end{center}
	\label{th:inter-edge-del}
\end{theorem}

\noindent
If $\comH \in \comS$, the (possible) modularity loss depends on 
the rank (in terms of degree) of $\comH$ in $\comS$. 
 If $\comH \notin \comS$, the result of Theorem~\ref{th:inter-edge-del} 
 still holds; the edge deletion with the possible highest loss is $(u,w)$: $u \in 
 \com_i\cap\comH$, $w \in \com_j$, with $i\neq j$, where the sum of 
 $\degC{\com_i}$ and $\degC{\com_j}$ is minimal.	We now consider the 
 deletion of an intra-community edge. 
	\begin{theorem} For any 
			intra-community edge deletion $(u,w)$: $u \in \com_i\cap\comH, w \in \com_i$ 
			giving 
			{\small 
				$\netN=(\nodes,\edges\setminus\{\e{u}{w}\})$} we have that:
				\begin{center}
						\small 
			$\modLoss>0$ if, and only if,  
						$\frac{m-\eta}{m(m-1)}+\frac{\delta(2m-1)-4m^2(\degC{C_i}-1)}{4m^2(m-1)^2}>0$
							\end{center}
			\label{th:intra-edge-del}
	\end{theorem}
	
	\noindent
	If $\comH=\com_i \in \comS$, then the possible modularity loss deriving 
	from an intra-community edge deletion depends on the degree of 
	 $\com_i$. 
 If $\comH \notin \comS$, Theorem~\ref{th:intra-edge-del} 
 still holds; the edge deletion with the possible highest loss is $(u,w)$: $u \in 
 \com_i\cap\comH, w \in \com_i$ and $\degC{\com_i}$ is minimal.
	By considering an inter-community edge deletion between the 
	communities $\com_i$ and $\com_j$ (and obtaining the network $\net'$) 
	and an intra-edge deletion in the community $\com_i$ (and obtaining the 
	network $\net''$) we have that:
	$\modularity{\net'}{\comS}-	
	{\modularity{\net''}{\comS}}$=$\frac{2\numedges-1+\degC{\com_j}-\degC{\com_i}}{2(\numedges-1)^2}$. Since $\degC{\com_i}\leq2m$, we have that $\modularity{\net'}{\comS}-	
			{\modularity{\net''}{\comS}}> 0$ and, thus,
		${\modularity{\net'}{\comS}}>\modularity{\net''}{\comS}$.
	\begin{corollary}
		The best edge deletion, in terms of possible modularity loss, is an intra-community edge in the	community $\com_i$ having the lowest degree and 
				 such that $\com_i\cap\comH\neq \emptyset$. The modularity 
				 loss is the same no matter the pair of nodes 
				 $u\in\com_i\cap\comH$ 
				and $w\in\com_i$.
		\label{cor:edge-deletion}
	\end{corollary}
\subsection{The $\algoMod$ algorithm}
\label{sec:modularity-algo}
The community deception algorithm $\algoMod$ outlined in 
Algorithm~\ref{algo-modularity-min} builds upon 
the analysis performed in Section~\ref{sec:impact-edges-modularity}.  
$\algoMod$ at each step compares the two most convenient edge updates 
(line 14) as per Corollary~\ref{cor:edge-addition} (lines 5-6; lines 10-12) and 
Corollary~\ref{cor:edge-deletion} (lines 3-4; lines 8-9) and returns the 
update giving the highest modularity loss. Since the loss only 
depends from the degree of communities the algorithm returns the best 
edge update by randomly picking its endpoints. 
\vspace{-.2cm}
\begin{algorithm}
	{\scriptsize
	\caption{\scriptsize $\algoMod$ - Community deception via Modularity}
	\label{algo-modularity-min}
	\begin{algorithmic}[1]
		\Procedure{getBestUpdateModularity}{\netD,$\comS$,$\comH$}
\If {$\comH \in \comS$}
\State Let $(n_k,n_l)$:$\{n_k,n_l\}\subseteq\comH$ with $(n_k,n_l)\in E$ and 
$n_k,n_l$ randomly selected  
\State Let $(n_p,n_t)$:$n_p$$\in$$\com_i$$\cap$$\comH$ 
$n_t$$\in$$\com_j$ randomly selected; $\com_i,\com_j$ highest 
degs; $(n_p,n_t)\notin E$
\Else
\State Let $(n_k,n_l)$: 
$n_k$$\in$$\comH$$\cap$$\com_i$ and $n_l$$\in$$\com_i$ randomly 
selected; $\com_i$ has lowest degree; 
$(n_k,n_l)\in E$
\State Let $n_p$$\in$$\comH$$\cap$$\com_i$ be randomly selected, $\com_i$ be the highest degree community 
\State  Let $n_t$$\in$$\com_j$, $\com_j\neq\com_i$ has the highest degree and $(n_p,n_t)\notin E$ 

\EndIf
\State $\modLoss_{del}$= intra-community edge deletion loss for $(n_k,n_l)$
computed according to Th.~\ref{th:intra-edge-del}
\State $\modLoss_{add}$= inter-community edge addition loss for $(n_p,n_t)$
computed according to Th.~\ref{th:inter-edge-add} 

\If {$\modLoss_{del} \geq \modLoss_{add}$}
\State \Return $(V,E\setminus\{(n_k,n_l)\})$ 
\Else
\State \Return $(V,E\cup\{(n_p,n_t)\})$  
\EndIf
		\EndProcedure
	\end{algorithmic}
}
\end{algorithm}
\section{Community Deception via Safeness}
\label{sec:safeness}
In this section we describe $\algoSaf$, our second algorithm for community 
deception. 
Differently from the $\algoMod$, this approach is independent from any 
cluster quality measures. We now introduce the notion of node safeness.
%
\begin{defn}\textbf{(Node Safeness)}.
	Let $\netD$ be a network, $\comH 
	\subseteq \nodes$ a community, and $u \in \comH$ a member of 
	$\comH$. The safeness of $u$ in $\net$ is defined as:
	\begin{equation}
	\safScoreNode{u}:=\frac{1}{2}\frac{|\reachNodes{u}{\comH}|-|E(u,\comH)|}
	{|\comH|-1}+\frac{1}{2}
	\frac{|\edgesNoComH{u}|}{\degC{u}}
	\label{eq:safeness-node}
	\end{equation}
	
	where $\reachNodes{u}{\comH}\subseteq\comH$ is
		the set of nodes reachable from $u$ passing only via nodes in 
		$\comH$,  $E(u,\comH)$ is the set of edges between $u$ and some 
		node in $\comH$, $\edgesNoComH{u}$ is the set of edges 
		between 
		$u$ 
		and some node not in $\comH$.
			\end{defn}

The leftmost part of eq.~(\ref{eq:safeness-node}) takes into account the 
portion of nodes in $\comH$ that can be reached only via other nodes in 
$\comH$ balanced by the number of intra-community edges. In the 
ideal situation a member of $\comH$ will be able to reach all the other 
members of $\comH$ with the minimum number of edges, that is, one. 
This gives an account of how-well $u$ can transmit information in 
$\comH$. 
The second term of eq. 
(\ref{eq:safeness-node}) gives an account on how $u$ is ``hidden" inside 
the network with respect to its degree. To increase its safeness 
$u$ should diversify its connections, that is, have the right 
proportion of links with members of communities other than $\comH$.
We now define the safeness of $\comH$ inside a network $\net$:
\begin{defn}(\textbf{Community Safeness Score}). Given a network $\netD$ 
	and a community $\comH\subseteq \nodes$, the safeness of 
	$\comH$ denoted by $\safScore$ is defined as:
	\begin{equation*}
		\safScore= \frac{\sum_{u \in 
		\comH}\safScoreNode{u}
		}{|\comH|}
		\label{eq:safeness-community-1}
	\end{equation*}
	
\end{defn}

Defining the safeness of $\comH$ starting from the safeness of its 
members allows to identify the {least safe} and rewire their links to 
increase the score of the whole $\comH$. Safeness allows to 
control different aspects of a community such as reachability and 
internal/external edge balance that were not taken into account by the 
modularity loss. $\comH's$ members should be able to 
communicate while at the same time diversify their connections 
with members outside $\comH$. Moreover, incorporating reachability in the 
safeness formula avoids to disconnect $\comH$, which can occur when 
using the modularity loss as shown in the example in 
Fig.~\ref{fig:running-example} (b) where node 24 was disconnected from 
the other members of $\comH$.

Our second deception algorithm $\algoSaf$ considers the 
function $\phi_{\comAlgo}$ to be the 
\textit{safeness gain} $\sIncrease$=$\safScoreN{\netN}-\safScore$ and thus 
the goal is to find the set of edge updates where $\sIncrease$ is maximized.
%
\subsection{Impact of Edge Updates on Safeness}
\label{sec:impact-edges-safeness}
As usual, we treat separately edge additions and 
deletions. However, note that given a node $u\in\comH$ the safeness score 
only considers the portion of edges incident to $u$ connecting $u$ to other 
members of $\comH$ and the portion of edges that connect $u$ to nodes 
not in $\comH$. Thus, instead of talking about intra-community and 
inter-community edges, we will talk about intra-$\comH$ and 
inter-$\comH$ edges. We assume wlog that for every inter-$\comH$ edge 
$(u,w)$ we have $u\in\comH$ and $w\notin\comH$. Let $\netD$ be a 
network  and $\comH\subseteq V$ a community having safeness 
$\safScore$, we have the following results (the proofs of the theorems are 
available in the Appendix.)

\smallskip
\noindent\underline{Edge Addition.}  We start with inter-$\comH$ edge 
additions. 
\begin{theorem}  For any inter-$\comH$ edge addition $(u,w)$ s.t. $u\in\comH$ and $v\notin\comH$ giving 
$\netN=(\nodes,\edges\cup\{\e{u}{w}\})$ we have that
$\sIncrease>0$.
\end{theorem}

Moreover, among all 
the possible inter-$\comH$ edge addition, the more beneficial is that 
performed by the node with the minimum ratio 
$\frac{|\edgesNoComH{u}|}{deg(u)}$.
We now analyze the case of the addition of an intra-$\comH$ edge.

\begin{theorem} For any intra-$\comH$ edge addition $(u,w)$ s.t. $\{u,w\}\subseteq \comH$ giving 
	$\netN=(\nodes,\edges\cup\{\e{u}{w}\})$ we have
	$\sIncrease>0$ if: \textit{(i)} $w\notin \reachNodes{u}{\comH}$ in $G$; \textit{(ii)} $w\in 
\reachNodes{u}{\comH}$ in $\netN$ and \textit{(iii)} the following condition 
holds:
{\scriptsize{
\begin{equation*}
\sum_{v\in\comH_u\setminus\{u\}}\hspace{-.1cm}\frac{|\comH_w|}{2(|\comH|\mbox{-}1)}+
\hspace{-.3cm}\sum_{v\in\comH_w\setminus\{w\}}
\hspace{-.2cm}\frac{|\comH_u|}{2(|\comH|\mbox{-}1)}+
\frac{|\comH_w|\mbox{-}1}{2(|\comH|\mbox{-}1)}+
\frac{|\comH_u|\mbox{-}1}{2(|\comH|\mbox{-}1)}-
\frac{|\edgesNoComH{u}|}{2\degC{u}(\degC{u}\mbox{+}1)}-
\frac{|\edgesNoComH{w})|}{2\degC{v}(\degC{v}\mbox{+}1)}>0
\end{equation*}}}
 where $\comH_u$ and $\comH_w$ are the two disconnected components 
 of $\comH$ in $G$ to which $u$ and $w$ belong before the addition of 
 $(u,w)$.
\end{theorem}

%
The above theorem deals with the addition of an intra-$\comH$ edge 
where both members belong to $\comH$. The possibility for such an edge 
to increase the safeness of the community occurs when it allows to connect 
previously disconnected portions of $\comH$. If no new communication 
paths among nodes of the community are made available, the new edge 
will certainly decrease the safeness score. Intuitively, this is justified by the 
fact that if the edge does not bring any advantage in terms of connectivity 
among the nodes of $\comH$, it will have only the effect to get $u$ and $w$ 
more connected to members of $\comH$; thus, it is likely that $u$ and $w$ 
will be considered part of the same community. 
We believe that, because of the notion of community itself, it is reasonable to 
consider that members of $\comH$ are reachable in $\net$ from one 
another via paths involving 
other members of $\comH$, and, thus, the induced subgraph of 
$G$ on the nodes in $\comH$ should have a single connected component. 
%
\vspace{-.1cm}
\begin{corollary}
The best addition is an inter-$\comH$ edge from $u\in\comH$ 
having the lowest ratio $\frac{|\edgesNoComH{u}|}{deg(u)}$. The safeness gain is 
the same for each edge (u,w) where $w$$\in$$V$$\setminus$$\comH$.
\label{cor:edge-addition-saf}
\end{corollary}
\vspace{-.35cm}
\noindent\underline{Edge Deletion.} We start with the deletion of an 
inter-$\comH$ edge.

\begin{theorem} For any inter-$\comH$ edge deletion $(u,w)$ such that 
$u \in \comH, w \notin \comH$ giving 
	$\netN=(\nodes,\edges\setminus\{\e{u}{w}\})$ we have that  
	$\sIncrease<0$.
\end{theorem}

We now analyze the case of the deletion of an intra-$\comH$ edge by 
showing that it does not always bring an increase of safeness.

\begin{theorem} For any intra-$\comH$ edge deletion $(u,w)$ s.t. $\{u,w\}\subseteq \comH$ giving 
	$\netN=(\nodes,\edges\setminus\{\e{u}{w}\})$ we have that
	$\sIncrease>0$ if:
	\vspace{-.3cm}
\begin{itemize}
\item $w\in \reachNodes{u}{\comH}$ in $\netN$; or
\item $w\notin \reachNodes{u}{\comH}$ in $\netN$ and it holds {\scriptsize$
\sum_{v\in\comH_u\setminus\{u\}}\hspace{-.cm}\frac{-|\comH_w|}{2(|\comH|\mbox{-}1)}+\sum_{v\in\comH_w\setminus\{w\}}\hspace{-.1cm}\frac{-|\comH_u|}{2(|\comH|\mbox{-}1)}-\frac{|\comH_w|\mbox{+}1}{2(|\comH|\mbox{-}1)}-\frac{|\comH_u|\mbox{+}1}{2(|\comH|\mbox{-}1)}+\frac{\edgesNoComH{u}}{2\degC{u}(\degC{u}\mbox{-}1)}+\frac{|\edgesNoComH{w}|}{2\degC{w}(\degC{w}\mbox{-}1)}$$>$$0$},
 where $\comH_u$ and $\comH_w$ are the two disconnected components of 
$\comH$ obtained after deleting $(u,w)$ to which $u$ and $w$ belong.
\end{itemize}
\end{theorem}

Similarly to the previous case, since $\comH$ is a community, it is 
reasonable to preserve the possibility for the members of $\comH$ to 
communicate with each other and thus that induced subgraph of $G'$ on the 
nodes in $\comH$ has a single connected component. By looking at the 
previous theorems, the following corollary holds.

\begin{corollary}
The best edge deletion is an intra-$\comH$ edge $(u,w)$ with $\{u,w\}\subset\comH$  
having the highest value {\small{
$\frac{1}{|\comH|-1}+\frac{|\edgesNoComH{u}|}{2\degC{u}(\degC{u}-1)}+\frac{|E(\edgesNoComH{w})|}{2\degC{w}(\degC{w}-1)}$}}. 
\label{cor:edge-deletion-saf}
\end{corollary}
\vspace{-.2cm}
\subsection{The $\algoSaf$ algorithm}
\label{sec:safeness-algo}
The community deception algorithm $\algoSaf$ outlined in 
Algorithm~\ref{algo-safeness-max} builds upon 
the analysis performed in Section~\ref{sec:impact-edges-safeness}.  
$\algoSaf$ at each step compares the two most convenient edge updates 
(line 13) as per Corollary~\ref{cor:edge-deletion-saf} (lines 2-3) and 
Corollary~\ref{cor:edge-addition-saf} (lines 5-7; lines 9-11) and returns the 
update giving the highest safeness gain. 
\vspace{-.1cm}
\begin{algorithm}
	{\scriptsize
	\caption{\scriptsize $\algoSaf$ - Community deception via Safeness}
	\label{algo-safeness-max}
	\begin{algorithmic}[1]
		\Procedure{getBestUpdateSafeness}{\net,$\comH$}
		\State Let $(n_k,n_l)$: $\{n_k,n_l\}\subseteq\comH$ and $n_k$ and 
		$n_l$ chosen according to Cor.~\ref{cor:edge-deletion-saf}. 
		\State $\sIncrease^{del}$= intra-$\comH$ edge deletion gain for $(n_k,n_l)$
\State Let $n_p$$\in$$\comH$ be the node having the lowest ratio $\frac{|\edgesNoComH{n_p}|}{deg(n_p)}$ 
\State Let $n_t$$\in$$V$$\setminus$$\comH$ be a randomly selected node, s.t.  $(n_p,n_t)$$\notin$$E$
\State $\sIncrease^{add}$= inter-$\comH$ edge addition gain for $(n_p,n_t)$

\If {$\sIncrease^{del} \geq \sIncrease^{add}$}
\State \Return $(V,E\setminus\{(n_k,n_l)\})$ 
\Else
\State \Return $(V,E\cup\{(n_p,n_t)\})$  
\EndIf
		\EndProcedure
	\end{algorithmic}
}
\end{algorithm}
\vspace{-.3cm}

%
%
\section{Experimental Evaluation}
\label{sec:evaluation}
We now report on the experimental evaluation. We start by describing the 
experimental setting, the evaluation methodology, and then report on the 
results.
%

%
%
\noindent
\textbf{Experimental Setting.}
We performed all the experiments in the worst-case 
scenario, that is $\comH\in \comS$; we pick a random $i\leq |\comS|$ and 
assume $\comH=\com_i$.
We investigated to what extent our algorithms $\algoMod$ and $\algoSaf$ 
are able to deceive the following community 
detection algorithms available 
in igraph\footnote{\url{http://igraph.org/r}}:
Louvain~\cite{blondel2008fast} (\texttt{louv}), 
Optimal~\cite{brandes2008modularity} 
(\texttt{opt}), InfoMap~\cite{rosvall2008maps} (\texttt{inf}), 
WalkTrap~\cite{pons2005computing} 
(\texttt{walk}),Greedy~\cite{clauset2004finding} (\texttt{gre}), 
SpinGlass~\cite{reichardt2006statistical} (\texttt{spin}), Label 
propagation~\cite{raghavan2007near} (\texttt{lab}), Leading 
Eigenvectors~\cite{newman2006finding} (\texttt{eig}), and 
Edge-Betweeness~\cite{girvan2002community}  (\texttt{btw}).

\smallskip
\noindent
\underline{Datasets.} We considered the following networks:
Zachary's Karate Club (\texttt{kar}), Dolphins 
(\texttt{dol}), Les Miserables (\texttt{lesm}), American 
College Football (\texttt{ftb}), Madrid Terrorist Network (\texttt{mad}), Books 
about US Politics (\texttt{pol}), and USA Power Grid 
(\texttt{pow}) available 
online\footnote{\url{http://www-personal.umich.edu/~mejn/netdata}}.
We also generated networks according to the 
community detection benchmark generator described by Lancichinetti and 
Fortunato~\cite{lancichinetti2009benchmarks}\footnote{The code is 
available 
at \url{https://sites.google.com/site/santofortunato/inthepress2}}.
Experiments have been conducted on a PC i5 CPU 2.6 GHz and 8GB RAM. 
The code of our implementation in R, the datasets and instructions about how to 
replicate the experiments are available 
online\footnote{\url{https://github.com/giuseppepirro/com-deception}}.

\smallskip
\noindent
\textbf{Evaluation Methodology.}  To measure the success of community 
deception algorithms we define the community deception score 
$\hScore{\net}$.
%
\vspace{-.2cm}
\begin{defn} \textbf{(Community Deception Score)}. Given the output of a 
	detection algorithm $\comAlgo$ $\comSD$, the community deception score 
$\hScore{\net}$ is:
\begin{equation}
\scriptsize
\hScore{\net}(\comH,\comAlgo)=(1-\frac{|S(\comH)|\mbox{-}1}{|\comH|\mbox{-}1})\left(\frac{1}{2}\frac{|\{C_i|C_i\cap\comH\neq\emptyset\}|\mbox{-}1}{|\comH|}
+\frac{1}{2}(1-\frac{\sum_{C_i\mid
 C_i\cap\comH\neq\emptyset}\frac{|C_i\cap\comH|}{|C_i|}}{|\{C_i\mid 
C_i\cap\comH\neq\emptyset\}|})\right)
\label{eq:deception-score}
\end{equation}
\vspace{-.2cm}
 $S(\comH)$ are the connected components in the subgraph induced 
by $\comH's$ members.
\label{def:deception-score}
\end{defn}

The first multiplicative factor in eq.~(\ref{eq:deception-score}) takes into 
account the fact that a deception 
algorithm should preserve as much as possible reachability between 
nodes in $\comH$. The best situation is when all nodes are in a 
single connected component while the worst case occurs when they all 
belong to a different connected component.
The second multiplicative factor includes two terms. The first term measures 
the \textit{community spread}, that is, how $\comH's$ members are spread 
within $\comS$. It reaches its maximum when each member of $\comH$ is 
placed by $\comAlgo$ in a different community. The second term measures 
\textit{community hiding}, that is, the average percentage of $\comH's$ 
members in the communities in $\comSD$. The ideal situation is when 
each community $\com_i\in\comS$ contains a little percentage of
$\comH's$ nodes. Summing up, $\hScore{\net}$$\sim$1 if 
\textit{(i)} $\comH's$ nodes are in a single connected components and 
\textit{(ii)} each such nodes belongs to a different (large) community. 
Conversely, $\hScore{\net}$=0 if \textit{(i)} each member of $\comH$  
belongs to a different connected component or \textit{(ii)} $\comH\in 
\comS$.
The evaluation has been conducted as shown in 
Algorithm~\ref{algo-evaluation-algorithms}. We 
consider a budget of changes $\budget$ such that $|\netPE|+|\netME|\leq \budget$ and compute
the new values of modularity, safeness and deception score
\textit{after} applying all the updates found by the deception algorithms and 
compare them with their initial values.
\vspace{-.7cm}
\begin{algorithm}[!h]
	{\scriptsize
		\caption{\scriptsize Evaluating Community Deception Algorithms}
		\label{algo-evaluation-algorithms}
		\begin{algorithmic}[1]
			\Procedure{evaluateDeceptionAlgo}{\net,$\beta$, 
			$\comAlgo$,$\algoProtGen$}
			\State $\comS$=$\comAlgo(\net)$
			\State $\comH$=getTargetCommunity($\comS$); 
			\State $\modCH$=initialMod($\comS,\net$); 
			$\safScore$=initialSafe($\comH,\net$); 
			$\hScore{\net}$=initialDecept($\comH,\comS$,\net)
					\While{$\beta > 0$}
			\State $E'$=		
		getBestUpdate($\net$,$\comS$,$\comH$,$\algoProtGen$) /* computed via 
		$\algoMod$ or $\algoSaf$ */
			\State $\netN$=$(V,E')$; $\budget$=$\budget$-1
			\EndWhile
			\State $\comSN$=$\comAlgo(\netN)$; 
			\State
			$\modularitySymbol_{\netN}(\comSN)$=finalMod($\comS',\netN$); 
			$\safScoreN{\netN}$=finalSafe($\comH,\netN$); s
			$\hScore{\netN}$=finalDecept($\comH,\comSN,\netN$)
			\EndProcedure
		\end{algorithmic}
	}
	\end{algorithm}
	\vspace{-1.3cm}
	%
\subsection{Evaluation Results} 
We start with \underline{real world 
networks}. Fig.~\ref{fig:heat-map} reports the values of the deception score 
(average of 10 runs) after applying our deception algorithms when varying the 
budget of updates 
$\budget$ from 1 to 4. Each column represents a dataset and each row a 
community detection algorithm. The range of the colors reflects the final 
value of the deception score (green is better). White cell reflect problems 
with the detection algorithms (e.g., 
\texttt{spin} does not work with disconnected networks, \texttt{opt} was 
stopped after 1h). 
As it can 
be observed, results vary with the network and detection/deception 
algorithm. 
A quick look suggests that deception based on safeness (i.e., $\algoSaf$) 
generally performs better. 
\vspace{-.7cm}
\begin{figure}[!h]
	\centering
	\includegraphics[width=\textwidth]{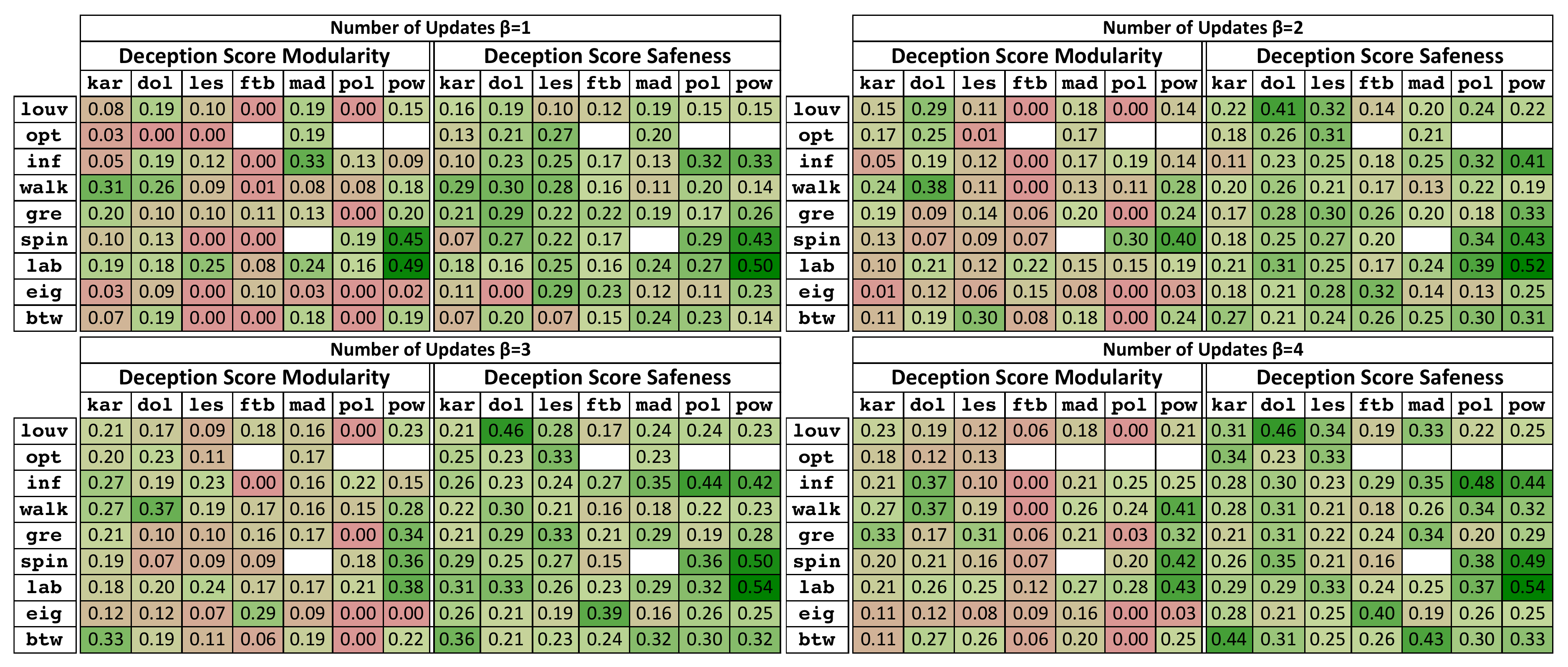}
		\vspace{-.7cm}
	\caption{Deception score ($\hScore{}$) for modularity-based and 
	safeness-based deception.}
	\label{fig:heat-map}
	\vspace{-.6cm}
\end{figure}

Moreover, the 
level of deception increases as the number of updates allowed increases for 
almost all the algorithms. When $\budget$=1, the deception 
algorithm based on modularity (i.e., $\algoMod$) obtains the worst 
deception values with the network \texttt{ftb}, which represents the 
schedule of football games between American college teams. On the same 
network $\algoSaf$ performs clearly better. Note also that $\algoMod$ gives 
the best deception score with $\budget$=1 for the network (\texttt{pow}), 
which represents the topology of the Wester USA power grid and the 
algorithms \texttt{spin} and \texttt{lab}. From a deception point of view, this 
means that 
these two algorithms are deceivable with only one update. In general, from 
Fig.~\ref{fig:heat-map} it can be observed that already with a single update 
($\budget$=1) 
safeness-based deception performs reasonably well, considering that our 
experiments are conducted in the worst case scenario (i.e., 
$\hScore{\net}$=0). We conducted further experiments (note reported for 
sake of space) by considering 
$\budget$=5 and $\budget$=6 and observed an increase of $\hScore{}$ for 
both modularity-based 
and safeness-based deception.
%
%
%
\vspace{-.869cm}
\begin{figure}[!h]
		\begin{minipage}{1.0\textwidth}
		\centering
		\includegraphics[width=\textwidth]{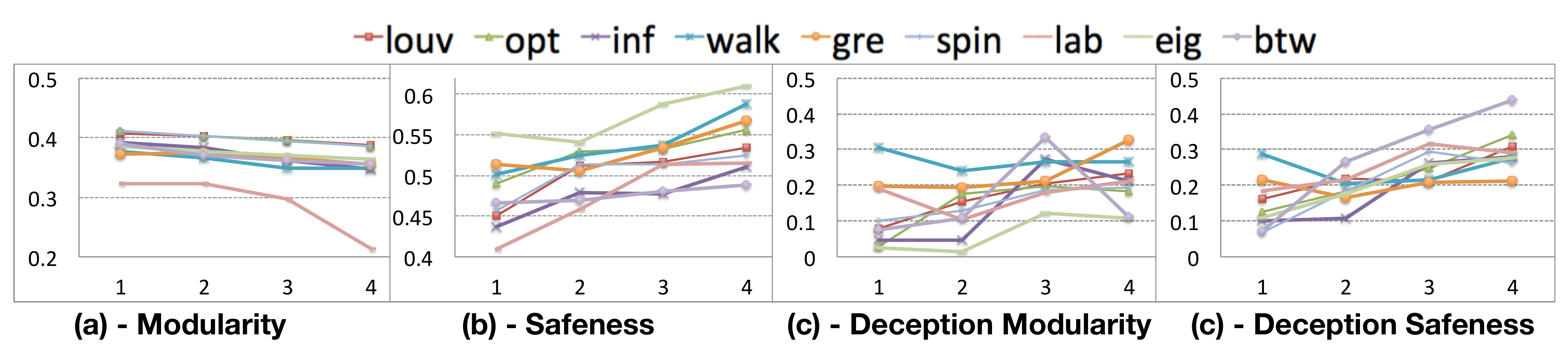}
		\vspace{-.85cm}
	\caption{Network \texttt{kar}: 34 nodes and 78 edges. 
		Avg $|\comS|$=4; Avg $|\comH|$=13.}
	\label{fig:all-kar}
	\end{minipage}
		\begin{minipage}{1.0\textwidth}
					\centering
		\includegraphics[width=\textwidth]{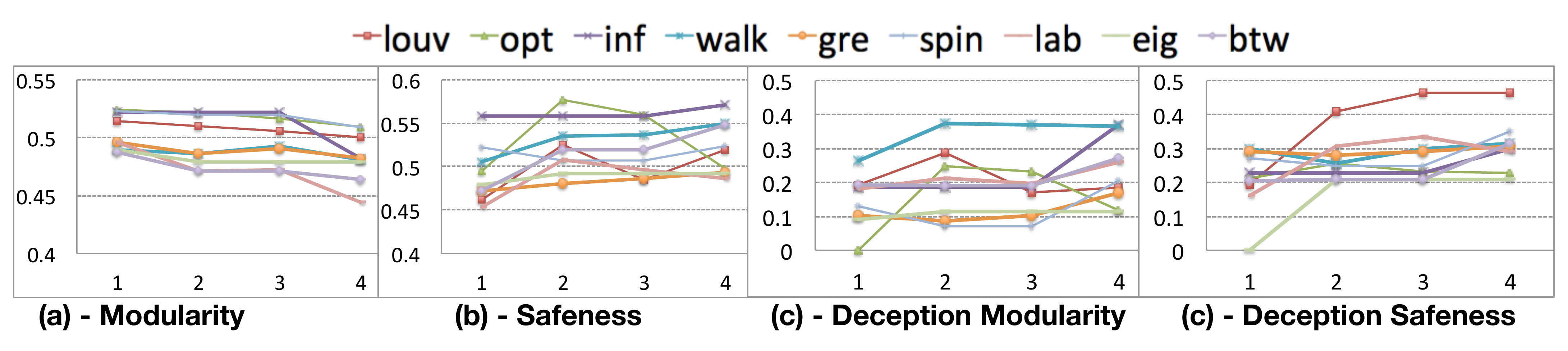}
		\vspace{-.85cm}
		\caption{Network \texttt{dolph}: 62 nodes and 159 edges. 
			Avg $|\comS|$=9; Avg $|\comH|$=11.}
		\label{fig:all-dolph}
	\end{minipage}
\begin{minipage}{1.0\textwidth}
	\centering
	\includegraphics[width=\textwidth]{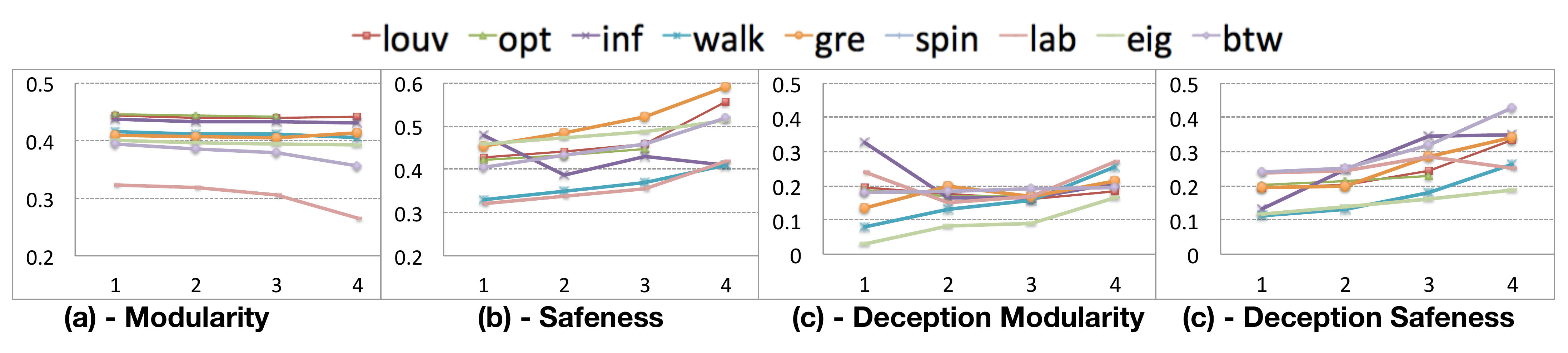}
		\vspace{-.85cm}
	\caption{Network \texttt{mad}: 62 nodes and 243 edges. 
		Avg $|\comS|$=6; Avg $|\comH|$=12.}
	\label{fig:all-mad}
\end{minipage}
\begin{minipage}{1.0\textwidth}
	\centering
	\includegraphics[width=\textwidth]{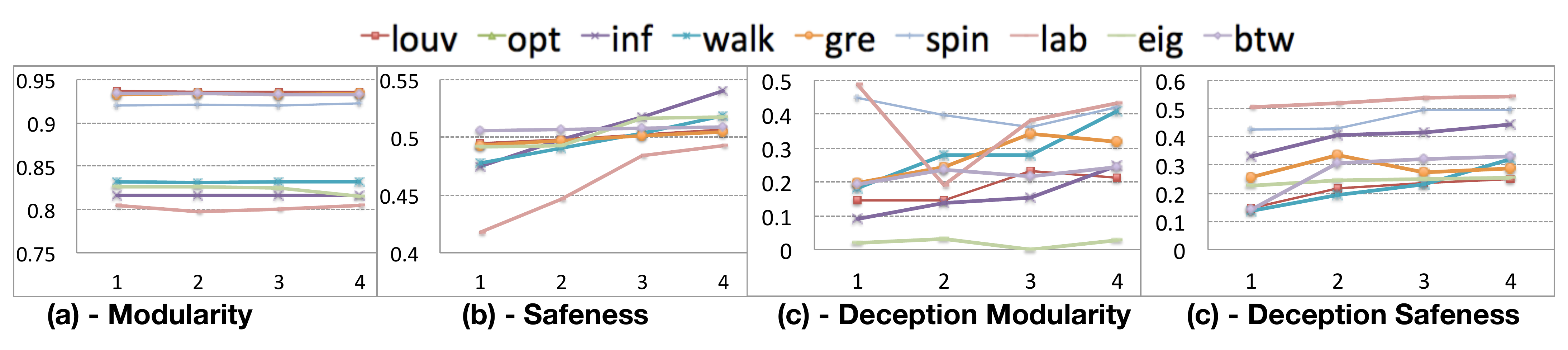}
		\vspace{-.85cm}
	\caption{Network \texttt{pow}: 6594 nodes and 4941 edges. 
		Avg $|\comS|$=40; Avg $|\comH|$=174.}
	\label{fig:all-pow}
\end{minipage}
\vspace{-.7cm}
\end{figure}

Figures (\ref{fig:all-kar})-(\ref{fig:all-pow}) provide a more detailed 
view of modularity, safeness and deception score for four of the considered 
networks.  We also report the size of the networks, average number of  
communities (Avg $|\comS|$) considering all the detection algorithms and 
average size of the community to hide (Avg $|\comH|$).
It can be noted that modularity decreases and safeness increases when the 
budget $\budget$ increases. This confirm the analyses performed in 
Section~\ref{sec:impact-edges-modularity} for modularity and 
Section~\ref{sec:impact-edges-safeness} for safeness. 
For modularity-based deception, $\hScore{}$ does not always 
increase when $\budget$ increases while for safeness-based deception 
$\hScore{}$ always increases. We explain this behavior by the fact that 
modularity-based deception simply aims at maximizing the modularity loss, 
while safeness also looks at reachability, which is in 
some form incorporated in the deception score (see 
Definition~\ref{def:deception-score}). Indeed, we observed that in some 
cases the modularity-based deception algorithm disconnects the community 
$\comH$  
while this is avoided by safeness-based deception. Figures 
(\ref{fig:all-kar})-(\ref{fig:all-pow}) also suggest that when the size of the 
network, the size of $\comH$ and the number of communities increases, 
the deception score is higher; 
it reaches the maximum value in the \texttt{pow} network (Avg$|\comH|$ is 
the 0.025\% of the size of $\net$) and the \texttt{lab} 
detection algorithm.

\smallskip
\noindent
We conducted experiments also on \underline{artificially generated 
networks}.The goal was to investigate the 
impact of the number of communities and size of $\comH$ on our deception 
algorithms because of the correlation observed in the experiments on real 
networks previously discussed. The community detection benchmark 
generator~\cite{lancichinetti2009benchmarks}
allows to generate networks having certain characteristics such as: size 
(\texttt{nodes}) average node degree (\texttt{avgD}), max degree 
(\texttt{maxD}), min size (\texttt{minC}) and max size (\texttt{maxC}) of the 
communities generated belonging to the ground truth. In this 
paper we are not interested in evaluating the performance of detection algorithms 
(viz. comparing their output with the ground truth). However, we noticed that the 
size of the communities found reflects pretty well the values of the parameters 
\texttt{minC} and \texttt{maxC} used in the generation of the networks. 
\vspace{-.8cm}
\begin{figure}[!h]
	\centering
	\includegraphics[width=\textwidth]{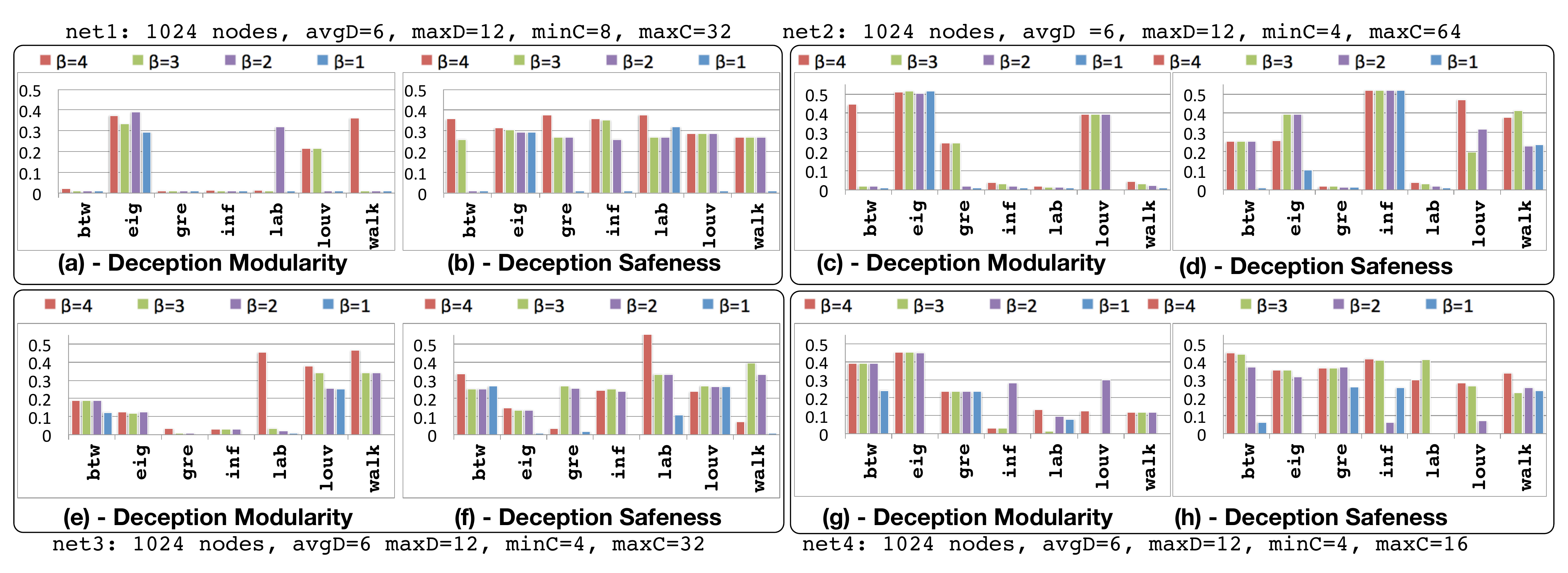}
	\vspace{-.7cm}
	\caption{Experiments on networks generated with the
		benchmarking software~\cite{lancichinetti2009benchmarks}.}
	\label{fig:exp-generated}
\end{figure}
\vspace{-.7cm}

We fixed the size and degree of nodes and generated networks having 
different community sizes. For sake fo space, we report in 
Fig.~\ref{fig:exp-generated} results on four of the ten generated networks. 
Moreover, we report the average results of 10 runs only for 
detection algorithms did not generate errors (e.g., the igraph 
implementation of \texttt{spin} and \texttt{opt} threw exceptions).
As our experiments are performed in the worst-case scenario (i.e., $\comH 
\in \comS$) we were able to investigate how a variation of the size of 
$\comH$ affects deception. It emerges 
(Fig~\ref{fig:exp-generated}) that 
when \texttt{maxC} decreases (i.e., \texttt{net4}) our deception algorithms 
are able to deceive a larger number of detection algorithms. We observed 
the same behavior in all the 10 networks.

\smallskip
\noindent
\textit{Summary}. In all the networks and for both 
deception algorithms we observed a dependency among size of $\comH$ 
(and $\net$), budget $\budget$ and deception score $\hScore{}$.  
When the size of $\comH$ increases by keeping constant $|\net|$ and 
$\budget$ the deception score decreases. This can be explained by the fact 
that spreading a larger number of nodes (as done by our deception 
algorithms) requires more network updates. In general, the lower the ratio 
$|\comH|/|\net|$ the higher $\hScore{}$ (no matter 
the detection algorithm).

We observed that safeness-based deception in $\sim$80\% 
of the cases does not change the number of communities while for 
modularity this happens in $\sim$60\% of the cases. We leave a 
more detailed study of this aspect for future work. As for the running 
times, they range from $\sim$1s to up to $\sim$15s (e.g., for the 
\texttt{pow} 
network); in general, safeness-based deception requires more time than 
modularity-based deception as it needs to check and preserve reachability 
among nodes in $\comH$.
\vspace{-.3cm}
\section{Concluding Remarks and Future Work}
\label{sec:conclusions}
So far the literature has focused on the design of community 
 detection algorithms. While this is certainly useful in some contexts where 
 one wants to understand the structure of a (complex) network, in some 
 others there is the need to hide the presence of a community.
In this paper we initiate the study of this problem that we dubbed as 
\textit{community deception}.
Our community deception algorithms are based on \textit{update rules} and 
thus suitable to 
deal with network dynamics~\cite{rozenshtein2014discovering}. Although 
we did not deal with node addition/deletions, it is immediate to see that a 
node addition corresponds to the creation of a node followed by (at least) an 
edge insertion, while a node deletion amounts at a set of edge deletions.
To measure the performance of deception algorithms we introduced the 
deception score $\hScore{}$. One may be tempted to devise algorithms that 
\textit{directly optimize} $\hScore{}$. $\hScore{}$ has been defined as a 
measure computed \textit{after} updating the network as 
suggested by the deception algorithms and recomputing the communities 
via detection algorithms. Our algorithms do 
not need to recompute communities for each update and thus provide a 
more efficient way to pursue community deception. From our experimental 
evaluation it emerged that the success of deception algorithms depends on 
the size of the community to be hidden, the total number of communities, 
and the size of $\net$.

Devising other instantiation of the general $\phi_{\comAlgo}$ function is 
an 
interesting line of future work. 
While we have studied how to \textit{deceive} detection 
algorithms, it is also interesting to investigate how detection algorithms can 
be made deception-aware. In this 
respect, a more speculative line of future work is to investigate whether 
certain types of complex networks such as 
biological networks exhibit some (natural) form of deceptive behavior. 
Another line of future work is to consider overlapping 
communities~\cite{galbrun2014overlapping} and networks with 
attributes~\cite{yang2013community}.
\vspace{-.4cm}
\bibliographystyle{abbrv}
\bibliography{biblio}
 \end{document}